\documentclass[orivec,oribibl]{llncs}

\usepackage[utf8]{inputenc}
\usepackage[T1]{fontenc}
\usepackage{lmodern}
\usepackage[english]{babel}
\usepackage{microtype}

\usepackage{gensymb} 

\usepackage{amsmath}
\usepackage{amsfonts}
\usepackage{amssymb}
\usepackage{amsthm} % don't load it when using LNCS style, unless using the relax trick

\usepackage{algorithm}
\usepackage{algpseudocode}
\usepackage{tikz}
\usepackage{pgfplots}
\usepackage{bbm}
\usepackage{multirow}
\usepackage{enumerate}
\usepackage{relsize}
\usepackage{mathtools}
\usepackage{wasysym}

\usepackage{complexity_llncs}

\usepackage{xcolor}
\usepackage{xifthen}
\usepackage{calc}
\usepackage{setspace}

\usepackage{complexity_llncs}

\usepackage{amssymb,latexsym}
\usepackage{tikz}
\usepackage{pgfplots}
\newcommand{\eqdef}{:=}
\pgfplotsset{compat=1.14}
\usepackage{float}

\usepackage{multicol}

\usepackage{caption}
\usepackage{color}
\usepackage[noend]{algcompatible}

\makeatletter
\newtheorem*{rep@theorem}{\rep@title}
\newcommand{\newreptheorem}[2]{%
	\newenvironment{rep#1}[1]{%
		\def\rep@title{#2 \ref{##1}}%
		\begin{rep@theorem}}%
		{\end{rep@theorem}}}
\makeatother

\newcommand{\norm}[1]{\left\lVert#1\right\rVert}

\renewcommand{\C}{\mathcal{C}}
\definecolor{coll}{HTML}{000090}

\renewcommand{\E}{\mathbb{E}}

\newcommand{\ie}{\textit{i.e.}}

\newcommand{\QRAM}{\mathrm{QRAM}}

\newcommand{\MNRS}{\mathrm{MNRS}}
\newcommand{\SETUP}{\mathrm{SETUP}}
\newcommand{\UPDATE}{\mathrm{UPDATE}}
\newcommand{\CHECK}{\mathrm{CHECK}}

\newcommand{\NBREP}{\mathrm{NB}_{\mathrm{REP}}}
\newcommand{\INIT}{\mathrm{INIT}}
\newcommand{\FAS}{\mathrm{FAS}}
\newcommand{\FINDALLSOLUTIONS}{\mathbf{FindAllSolutions}}

\newcommand{\Nd}{N}

\newcommand{\Sol}{\mathrm{Sol}}

\renewcommand{\COMMENT}[1]{}

\newcommand{\ket}[1]{|#1\rangle}

\newcommand{\inp}[2]{\langle{#1}|{#2}\rangle} % inproduct, < | >
\renewcommand{\inp}[2]{\langle{#1},{#2}\rangle} % inproduct, < | >
\def\01{\{0,1\}}
\def\01{\{0,1\}}

\newcommand{\eps}{\varepsilon}

\newcommand{\zo}{\{0,1\}}

\newcommand{\cv}{\vec{s}}

\newcommand{\sv}{\vec{s}}
\newcommand{\tv}{\vec{t}}

\newcommand{\vv}{\vec{v}}
\newcommand{\wv}{\vec{w}}
\newcommand{\xv}{\vec{x}}
\newcommand{\yv}{\vec{y}}

\newcommand{\hh}{\mathcal{H}}

\renewcommand{\S}{\mathbf{\mathcal{S}}}
\newcommand{\U}{\mathbf{\mathcal{U}}}
\renewcommand{\C}{\mathbf{\mathcal{C}}}
\newcommand{\V}{\mathbf{\mathcal{V}}} 
\renewcommand{\W}{\mathbf{\mathcal{W}}}

\newcommand{\scdot}{{\mkern 1.5mu\cdot\mkern 1.5mu}}

\newcommand{\cadre}[1]
{
	\begin{tabular}{|p{0.98\textwidth}|}
		\hline
		\vspace*{-0.3cm}
		#1 \\
		\hline
	\end{tabular}

}

\title{Lattice sieving via quantum random walks}
\author{ 
	Andr\'e Chailloux\inst{} \and Johanna Loyer\inst{}}
\institute{Inria de Paris, EPI COSMIQ, \\
	\email{andre.chailloux@inria.fr} {$\|$} 
	\email{johanna.loyer@inria.fr}}
\date{}
\begin{document}
\maketitle

\begin{abstract}
	Lattice-based cryptography is one of the leading proposals for post-quantum cryptography.
	The Shortest Vector Problem (SVP) is arguably the most important problem for the cryptanalysis of lattice-based cryptography, and many lattice-based schemes have security claims based on its hardness. The best quantum algorithm for the SVP is due to Laarhoven~\cite{Laa16} and runs in (heuristic) time $2^{0.2653d + o(d)}$. 
	In this article, we present an improvement over Laarhoven's result and present an algorithm that has a (heuristic) running time of $2^{0.2570 d + o(d)}$ where $d$ is the lattice dimension. We also present time-memory trade-offs where we quantify the amount of quantum memory and quantum random access memory of our algorithm. The core idea is to replace Grover's algorithm used in~\cite{Laa16} in a key part of the sieving algorithm by a quantum random walk in which we add a layer of local sensitive filtering. 
\end{abstract}

\section{Introduction}

Lattice-based cryptography is one of the most appealing modern public-key cryptography. It has worst case to average case reductions~\cite{Ajt96}, efficient schemes and allows more advanced primitives such as fully homomorphic encryption \cite{Gen09}. Another important aspect is that lattice based problems are believed to be hard even for quantum computers. Lattice-based cryptography is therefore at the forefront of post-quantum cryptography, especially in the NIST post-quantum standardization process. It is therefore very important to put a large effort on quantum cryptanalysis and to understand the quantum hardness of lattice problems in order to increase our trust in these post-quantum solutions.

For a given lattice $\mathcal{L}$, the Shortest Vector Problem (SVP) asks to find a short vector of this lattice. Solving the SVP is arguably the most important problem for the cryptanalysis of lattice-based cryptography. Additionally to its own importance, it is used as a subroutine in the BKZ algorithm, which is often the best attack on lattice-based schemes. There are two main families of algorithms for SVP: enumeration algorithms which are asymptotically slow but have small memory requirements, and sieving algorithm which have the best asymptotic complexities but have large memory requirements. For finding very small vectors, which is required by the BKZ algorithm, sieving algorithms are currently the most efficient algorithms despite their large memory requirements. Indeed, in the SVP challenge, the $10$ top performances are done by sieving algorithms and the current record solves SVP for $d = 180$\footnote{The SVP challenge can be accessed here \url{https://www.latticechallenge.org/svp-challenge}.}.

For a lattice $\mathcal{L}$ of dimension $d$, sieving algorithms solve SVP classically in time $2^{0.292d + o(d)}$ (with a heuristic analysis) using the local filtering technique introduced in \cite{BDGL16}. Laarhoven presented a quantum equivalent of this algorithm that runs in time $2^{0.265d + o(d)}$ while using as much space as in the classical setting, namely $2^{0.208d + o(d)}$. 
The BKZ algorithm is the most efficient known attack against all lattice-based schemes which were chosen at the third round of NIST standardization process\footnote{At this stage, there are $3$ encryption schemes / key encapsulation mechanisms: KYBER, NTRU and SABER as well as two signature schemes: DILITHIUM and FALCON.}. These two exponents are used for determining the number of bits of security in all these schemes hence improving the time exponent for SVP has direct implications on the security claims of these schemes. 
 
 \paragraph{Related work.}
Heuristic sieving algorithms were first introduced by Nguyen and Vidick \cite{NV08} that presented an algorithm running in time $2^{0.415d + o(d)}$ and using $2^{0.2075d}$ memory. A more efficient sieve in practice but with the same asymptotic running time was presented in \cite{MV10}. Then, there has been improvements by considering $k$-sieve algorithms \cite{WLTB11,ZPH14,Laa16}. Also, several works showed how to use nearest neighbor search to improve sieving algorithms \cite{LdW15,Laa15,BL16}. The best algorithm \cite{BDGL16} runs in time $2^{0.292d + o(d)}$ and uses locality-sensitive filtering.

In the quantum setting, quantum analogues of the main algorithms for sieving were studied \cite{LMvdP15,Laa16}. The best algorithm runs in time $2^{0.265d + o(d)}$ and is the quantum analogue of \cite{BDGL16}. There has been two more recent works on quantum sieving algorithms. First, quantum variants of the $k$-sieve were studied in \cite{KMPM19}, giving interesting time-space trade-off and a recent article \cite{AGPS20} studied more practical speedups of these quantum algorithms, {\ie} when do these gains in the exponent actually translate to quantum speedups. 

\paragraph{Contributions.}
In this article, we study and improve the asymptotic complexity of quantum sieving algorithm for SVP. This is the first improvement on the asymptotic running time of quantum sieving algorithms since the work of Laarhoven \cite{Laa16}\footnote{We are talking here only about the asymptotic running time, there are other metrics of interest that have been covered in \cite{KMPM19,AGPS20} where there were some improvements.}.

It is not \textit{a priori} clear how to use quantum random walks to adapt the algorithm from \cite{BDGL16}. This algorithm is divided into a pre-processing phase and a query phase. In this query phase, we have several points that are in a filter $F$, which means here that there are close to a specific point. We are then given a new point $\vv$ and we want to know whether there exists a point $\wv \in F$ such that $\norm{\vv \pm \wv}$ is smaller than $\min\{\norm{\vv},\norm{\wv}\}$\footnote{We remain a bit imprecise and informal here as we haven't properly described sieving algorithms yet.}. Then we do not know how to do better here than Grover's algorithm, which takes time $\sqrt{|F|}$. On the other hand, if instead of this query framework, we start from a filter $F$ and we want to find all the pairs $\vv,\wv$, then we can apply a quantum random walk. 

Even within this framework, there are many ways of constructing quantum random walks and most of them do not give speedups over \cite{Laa16}. What we show is that by adding proper additional information in the vertices of the random walk, in particular by adding another layer of filters within the vertices of the graph on which we perform the quantum walk, we can actually get some improvement over Grover's algorithm and achieve our speedups. 

We now state our results. We present here not only the running time but also the amount of classical memory, quantum memory and quantum RAM (QRAM) operations required. %We separate classical and quantum memory requirements since the former is and probably will be much cheaper and significantly easier to maintain than the latter. 

Our main theorem is an improvement of the best asymptotic quantum heuristic running time for the SVP bringing down the asymptotic running time from $2^{0.2653d + o(d)}$ to $ 2^{0.2570 d + o(d)}$. Our results are in the QRAM model where QRAM operations can be done efficiently. Notice that Laarhoven's result is in this model so our result are directly comparable to his. 

\begin{theorem}\label{Theorem:Main}
    There exists a quantum algorithm using quantum random walks that solves the SVP on dimension $d$ which 
    heuristically solves SVP on dimension $d$ in 
    time $ 2^{0.2570 d + o(d)}$, 
    uses QRAM of maximum size $ 2^{0.0767d}$,
    a quantum memory of size $2^{0.0495d}$ and a classical memory of size $\poly(d) \cdot 2^{0.2075d}$.
\end{theorem}
We can see that additionally to improving the best asymptotic running time, this algorithm uses much less quantum resources (both quantum memory and quantum RAM) than its running time which makes it fairly practical. We also present two trade-offs: a quantum memory-time trade-off and a QRAM-time trade-off. For a fixed amount of quantum memory, our algorithm performs as follows. 
\begin{theorem}[Trade-off for fixed quantum memory] 
	There exists a quantum algorithm using quantum random walks that solves the SVP on dimension $d$ which, for a parameter $M \in [0,0.0495]$, heuristically runs in time $2^{\tau_M d + o(d)}$, uses QRAM of maximum size $2^{\gamma_M d}$ and quantum memory of size $2^{\mu_M d}$ and a classical memory of size $\poly(d)2^{0.2075d}$ where 
	$$ \tau_M \in 0.2653 - 0.1670M + [-2\scdot10^{-5} ; 4\scdot10^{-5}]$$ 
	$$\gamma_M \in 0.0578 + 0.3829M - [0 ; 2\scdot10^{-4}] \quad ; \quad \mu_M = M.$$
\end{theorem}

 With this theorem, we obtain for $M = 0$ the quantum running time of Laarhoven's quantum algorithm and, for $M = 0.0495$, the result of Theorem \ref{Theorem:Main}.

We now present our second trade-off theorem where we fix the amount of QRAM.

\begin{theorem}[Trade-off for fixed QRAM]
	There exists a quantum algorithm using quantum random walks that solves SVP on dimension $d$ which for a parameter $M' \in [0,0.0767]$ heuristically runs in time $2^{\tau_{M'} d + o(d)}$, uses QRAM of maximum size $2^{\gamma_{M'} d}$, a quantum memory of size $2^{\mu_{M'} d}$ and uses a classical memory of size $\poly(d) \cdot 2^{0.2075d}$ where 
$$ \tau_{M'} \in 0.2925 - 0.4647M' - [0; 6 \scdot 10^{-4}] \quad ; \quad \gamma_{M'} = M'$$
$$\mu_{M'} \in \max\{2.6356(M'-0.0579),0\} + [0 ; 9\scdot10^{-4}].$$ 
\end{theorem}

 With this theorem, we obtain for $M' = 0$, the best classical exponent of \cite{BDGL16} (we can actually show the algorithm uses no quantum resources in this case). For $M' = 0.0577$, we retrieve Laarhoven's quantum exponent and for $M' = 0.0767$, we get Theorem \ref{Theorem:Main}.

This theorem can also be helpful if we want to optimize other performance measures. For example, it has been argued that having efficient QRAM operations is too strong and that performing a QRAM operation should require time at least $r^{1/3}$ where $r$ is the number of QRAM registers. This means we want to minimize the quantity $\lambda = \tau_{M'} + \frac{1}{3}\max\{\gamma_{M'},\mu_{M'}\}$\footnote{We want to minimize $2^{\tau_{M'}d}\cdot \left(2^{\gamma_{M'}d} + 2^{\mu_{M'}d}\right)^{1/3}$ which asymptotically is equivalent to minimizing $\tau_{M'} + \frac{1}{3}\max\{\gamma_{M'},\mu_{M'}\}$.}. We also show some mild improvements in this metric: the previous best known bound was $\lambda = 0.2849 $ \cite{Laa16,AGPS20} while using our theorem, we can retrieve the previous results by taking $M' = 0.0577$ and slightly improve it by taking ${M'} = 0.0767$ to obtain $\lambda = 0.2824$.

\paragraph{Organisation of the paper.} In section \ref{Section:QuantumPreliminaries}, we present preliminaries on Quantum computing. In Section \ref{Section:LatticePreliminaries}, we then present sieving algorithm, as well as useful statements on lattices. In Section \ref{Section:AlgorithmSkeleton}, we present the framework we use for sieving algorithm that we use and perform a first study of its time complexity. Next, we present in Section \ref{Section:QuantumWalk} the quantum walk that will allow our time improvements and in Section \ref{Section:Optimal1} the numerical values we achieve and the space-time trade-offs. We perform a final discussion in Section \ref{Section:Conclusion} and talk about parallelization of our algorithm as well as possible improvements. 
\section{Quantum computing preliminaries}\label{Section:QuantumPreliminaries}
\subsection{Quantum circuits.}
We consider here quantum circuits consisting of $1$ and $2$ qubit gate, without any locality constraint, meaning that we can apply a $2$ qubit gate from a universal set of gates to any pair of qubits in time $1$\footnote{We are only interested in asymptotic running time here so we are not interested in the choice of this universal gate set, as they are all essentially equivalent from the Solovay-Kitaev theorem (see \cite{NC00}, Appendix 3).}. We use the textbook gate model where the running time of a quantum circuit is just the number of gates used. The width of a circuit is the number of qubits it operates on, including the ancilla qubits. This quantity is important as it represents the number of qubits that have to be manipulated simultaneously and coherently. We will also call this quantity quantum memory. %Another quantity of interest is the depth of the circuit, which is important as a high depth quantum circuit will be much harder to achieve because of decoherence. 
%Moreover, while the number of gates is a good measure of time when each gate has to be computed separately, the depth is the good measure of time when we can perform many gates at the same time, typically when we have multiple quantum processors. 

When we will know much more precisely how quantum architectures look like, it will be possible to make these models more precise and replace the gate model with something more adequate. The gate model is still the most widely used in the scientific community and is very practical to compare different algorithms. We will use the gate model as our main model for computing quantum times but we will also include other interesting quantum figures of merit, such as quantum memory or Quantum Random Access Memory usage.

\subsection{Quantum Random Access Memory.} \label{section:QRAM} Quantum Random Access Memory (denoted hereafter QRAM) is a type of quantum operation which is not captured by the circuit model. Consider $N$ registers $x_1,\dots,x_N \in \zo^d$ stored in memory. A QRAM operation consists of applying the following unitary 
$$ U_{\QRAM} : \ket{i}\ket{y} \rightarrow \ket{i}\ket{x_i \oplus y}.$$
We say that we are in the $\QRAM$ model if the above unitary can be constructed efficiently, typically in time $O(d + \log(N))$. We can distinguish two different types of $\QRAM$: QRACM, where the registers $x_1,\dots,x_N$ are stored in some classical memory and QRAQM where the unitary $U_{QRAM}$ can be applied on fully quantum registers. The former being of course easier to achieve than the latter. 

$\QRAM$ operations are theoretically allowed by the laws of quantum mechanics and there are some proposals for building efficiently $\QRAM$ operations, such as \cite{GLM08}, even though its robustness has been challenged in \cite{AGJ+15}. The truth is that it very premature to know whether $\QRAM$ operations will be efficiently available in quantum computers. This would definitely require a major hardware breakthrough but as does quantum computing in general. 

While our results are mainly in the QRAM model, we will also discuss other metrics where the cost of a QRAM operation is not logarithmic in $N$ but has a cost of $N^x$ for a constant $x$.

\subsection{Grover algorithm.}
One formulation of Grover's search problem \cite{Gro96} is the following. 
We are given a list of data $x_1, ..., x_r$, with $x_i \in {E}$. 
Given a function $f : {E} \rightarrow \{0,1\}$, the goal is to find an $i$ such that $f(x_i)=1$, and to output "no solution" if there are no such $i$. Let $\Sol = \{i \in [r]: f(x_i) = 1\}$.

Classically, we cannot solve this problem with a better average complexity than $\Theta(\frac{r}{\Sol})$ queries, which is done by examining random $x_i$ one by one until we find one whose image is $1$ through $f$.
Quantum computing allows a better complexity. Grover’s algorithm solves this search problem in $O\left(\sqrt{\frac{r}{\Sol}}\right)$ queries to $f$. Applying Grover's algorithm this way requires efficient QRAM access to the data $x_1,\dots,x_r$.
%It requires to know the number $t$ of such $i$. In our case, we know that there will be in average one solution. 

%In our pseudo-code, we use the notation $\textbf{Grover}({E}, f)$, that returns a vector $\vv$ in the set ${E}$ such that $f(\vv) = 1$. 

\subsection{Quantum random walks.}
We present here briefly quantum random walks (QRW). There are several variants of QRW and we will use the $\MNRS$ framework, first presented in \cite{MNRS11}.

We start from a graph $G = (V,E)$ where $V$ is the set of vertices and $E \subseteq V \times V$ is the set of edges. We do not allow self loops which means that $\forall x \in V, (x,x) \notin E$ and the graph will be undirected so $(x,y) \in E \Rightarrow (y,x) \in E$. Let also $N(x) = \{y : (x,y) \in E\}$ be the set of neighbors of $x$. We have a set $M \subseteq V$ of marked elements and the goal of a QRW is to find $v \in M$.

Let $\eps = \frac{|M|}{|V|}$ be the fraction of marked vertices and let $\delta$ be the spectral gap of $G$\footnote{For a regular graph, if $\lambda_1 > \dots > \lambda_{|V|}$ are the eigenvalues of the normalized adjacency matrix of $G$, then $\delta = \lambda_1 - \max_{i = 2 \dots n} |\lambda_i|$.}. For any vertex $x$, we define $\ket{p_x} = \sum_{y \in N(x)} \frac{1}{\sqrt{|N(x)|}} \ket{y}$. We also define $\ket{U} = \frac{1}{\sqrt{|V|}} \sum_{x \in V} \ket{x}\ket{p_x}$. We now define the following quantities:
\begin{itemize}
	\item $\SETUP$ cost $\S$: the $\SETUP$ cost $\S$ is the cost of constructing $\ket{U}$.
	\item $\UPDATE$ cost $\U$: here, it is the cost of constructing the unitary 
	$$U_{\UPDATE} : \ket{x}\ket{0} \rightarrow \ket{x}\ket{p_x}. $$
	\item $\CHECK$ cost $\C$: it is the cost of computing the function $f_{\CHECK} : V \rightarrow \zo$ where $f_{\CHECK}(v) = 1 \Leftrightarrow v \in M$. 
\end{itemize}

\begin{proposition}\cite{MNRS11} \label{qwalk}
	There exists a quantum random walk algorithm that finds a marked element $v \in M$ in time 
	$$ \S + \frac{1}{\sqrt{\eps}}\left(\frac{1}{\sqrt{\delta}}\U + \C\right).$$
\end{proposition}

In order to compute the update cost, we can actually compute the classical running time of going from one vertex to another {\ie } starting from a vertex $x$ and constructing a vertex $y$ for a random neighbor $y \in N(x)$. Then, we can use this procedure in quantum superposition to construct the unitary $U_{\UPDATE}$. We refer to \cite{Amb07,MNRS11,deW19} for more details on these QRW.

\subsubsection{Quantum random walks on the Johnson graph.}
A very standard graph on which we can perform QRW is the Johnson graph $J(n,r)$. Each vertex $v$ consists of $r$ different (unordered) points $x_1,\dots,x_r \in [n]$ as well as some additional data $D(v)$ that depends on the QRW we want to perform. 

$v = (x_1,\dots,x_r,D(v))$ and $v' = (x'_1,\dots,x'_r,D(v'))$ form an edge in $J(n,r)$ iff. we can go from $(x_1,\dots,x_r)$ to $(x'_1,\dots,x'_r)$ by removing exactly one value and then adding one value. The Johnson graph $J(n,r)$ has spectral gap $\delta = \frac{n}{r(n-r)} \approx \frac{1}{r}$ when $r \ll n$~\cite{deW19}.

The additional data $D(v)$ here is used to reduce the checking time $\C$ with the drawback that it will increase the update time $\U$. Johnson graphs were often used, for example when trying to solve the element distinctness problem \cite{Amb07}, but also for the subset-sum problem \cite{BJLM13,HM18,BBSS20} or for code-based problems \cite{KT17}.

\subsubsection{Quantum data structures.}

A time analysis of quantum random walks on the Johnson graph was done in \cite{Amb07} when studying the element distinctness problem. There, Ambainis presented a quantum data structure that uses efficient QRAQM that allows in particular insertion and deletion in $O(\log(n))$ time where $n$ is the database size while maintaining this database in quantum superposition. Another paper on quantum algorithm for the subset problem using quantum random walks \cite{BJLM13} also presents a detailed analysis of a quantum data structure based on radix trees to perform efficient insertion and deletion in quantum superposition. All of these data structures require as much QRAQM registers as the number of registers to store the whole database and this running time holds only in the QRAM model. In our work, we will use such a quantum data structure and refer to the above two papers for explicit details on how to construct such quantum data structures.

\section{Lattice preliminaries}\label{Section:LatticePreliminaries}
\paragraph{Notations.}
The norm $\norm{\cdot}$ we use throughout this paper is the Euclidian norm, so for a vector $\vv = (v_1,\dots,v_d) \in \mathbb{R}^d$, $\|\vv\| = \sqrt{\sum_{i = 1}^d v_i^2}$. The inner product of $\vv = (v_1,\dots,v_d)$ and $\wv = (w_1,\dots,w_d)$ is $\langle \vv, \wv \rangle \eqdef \sum_{i = 1}^d v_i w_i$. 
The non-oriented angle between $\vv$ and $\wv$ is denoted $\theta(\vv, \wv) \eqdef \arccos\left(\frac{\inp{\vv}{\wv}}{\norm{\vv}\norm{\wv}}\right)$. 
%We use the usual Landau notations %$O, \widetilde{O}, \Omega$
We denote the $d$-dimensional sphere of radius $R$ by $\mathcal{S}^{d-1}_R := \{\vv \in \mathbb{R}^d : \|\vv\|=R \}$, and $S^{d-1} := S^{d-1}_1$. Throughout the paper, for a known integer $d$, we will write $N \eqdef (\sqrt{4/3})^{d}$.

\paragraph{Lattices.}
The $d$-dimensional lattice $\mathcal{L} \subset \mathbb{R}^m$ generated by the basis $B = (b_1, ... , b_n)$ with $\forall i, b_i \in \mathbb{R}^m$ is the set of all integer linear combinations of its basis vectors:
$ \mathcal{L}(B) = \Big\{ \sum_{i=1}^{d} \lambda_i ~ b_i , ~ \lambda_i \in \mathbb{Z} \Big\}$. 

\paragraph{Shortest Vector Problem.} 
Given a basis of a lattice $\mathcal{L}$, the Shortest Vector Problem (SVP) asks to find a non-zero vector in $\mathcal{L}$ of minimal norm. 
SVP is known to be NP-hard \cite{Ajt98}. This problem and its derivatives (SIS, LWE) have been used in several public-key cryptosystems, specifically as candidate for quantum-resistant cryptography \cite{dilithium, falcon, NTRU}.
Thereby, one of the most important ways to know their security and choose parameters is to estimate the computational hardness of the best SVP-solving algorithms.

\subsubsection{Sieving algorithms.}
\paragraph{SVP solving methods.}
The algorithm LLL \cite{LLL82} returns a reduced basis of a lattice in a polynomial time. However it is not sufficient to solve SVP.
All the fastest known algorithms to solve SVP run in exponential time. 
A first method is enumeration \cite{Kan83}, that solves deterministically SVP using low space but in super-exponential time in the lattice dimension $d$. 

Another method, which will interest us in this article, is lattice sieving \cite{NV08, MV10}. 
They are heuristic algorithms that probably solve SVP in time and space $2^{\Omega(d)}$. 
To this day, the best complexity for sieving in the QRAM model is obtained by quantum hypercone LSF \cite{Laa16} in $2^{0.2653d + o(d)}$ time and $2^{0.2075d + o(d)}$ space. 
Another algorithm \cite{KMPM19} uses $k$-lists to solve SVP in $2^{0.2989d + o(d)}$ time and $2^{0.1395d + o(d)}$ space. 

\paragraph{The NV-sieve.}
The NV-sieve \cite{NV08} is a heuristic algorithm. It starts with a list of lattice vectors, that we can consider of norm at most $1$ by normalization. Given this list and a constant $\gamma <1$, the NV-sieve returns a list of lattice vectors of norm at most $\gamma$. 
It iteratively builds lists of shorter lattice vectors by applying a sieve. This sieve step consists in computing all the sums (plus and minus) of two list vectors, and fills the output list with those which have norm at most $\gamma$. 
For $\gamma$ tending to $1$, two vectors form a reducing pair - {\ie} their sum is of norm at most $\gamma$ - iff. they are of angle at most $\pi/3$. 
The first list of lattice vectors can be sampled with Klein's algorithm \cite{Kle00} for example. A list size of $N^{1 + o(1)} = (\sqrt{4/3})^{d+o(d)}$ suffices to have about one reducing vector in the list for each list vector, as stated in \cite{NV08}. 
Because of the norms of the list vectors reduces with a factor by $\gamma < 1$ at each application of the algorithm, the output list will hopefully contain a non-zero shortest lattice vector after a polynomial number of application of the NV-sieve.

\paragraph{NNS and application to lattice sieving.}
A logic improvement of this algorithm is to use Neighbor Nearest Search (NNS) \cite{IM98} techniques. 
The NNS problem is: given a list $L$ of vectors, preprocess $L$ such that one can efficiently find the nearest vector in $L$ to a target vector given later. 
Used in the NV-sieve, the preprocessing partitions the input list in several buckets of lattice points, each bucket being associated with a hash function. The algorithm will only sum vectors from a same bucket, which are near to each other, instead of trying all pairs of vectors. 

\paragraph{Locality-sensitive hashing (LSH).}
A method to solve NNS is locality-sensitive hashing (LSH) \cite{IM98}. An LSH function is a hash-function that have high probability to collide for two elements if they are close, and a low one if they are far. 
Several categories of LSH functions exists: hyperplane LSH \cite{Cha02}, hypercone or spherical LSH \cite{AINR14, AR15} and cross-polytope LSH \cite{TT07}.

\subsubsection{Locality-sensitive filtering (LSF).}
\paragraph{Locality-sensitive filtering (LSF).}
More recently, \cite{BDGL16} improved NNS solving by introducing locality-sensitive filtering (LSF). 
LSF functions, called filters, map a vector $\vv$ to a boolean value: $1$ if $\vv$ survives the filter, and $0$ otherwise. They act similarly to LSH but only few vectors survive the filter. 

These filters are instantiated by hypercone filters, characterized by a vector $\sv$ and an angle $\alpha \in [0,\pi/2]$. 
For a filter $f$ of center $\sv$ and angle $\alpha$, the vector $\vv$ survives $f$ iff $\theta(\vv, \sv) \leqslant \alpha$. 
In this case, the filter $f$ is said relevant for $\vv$. The set of the vectors from a list that survives a filter $f$ is called its bucket, and is denoted $f_\alpha(\sv)$. More formally, we define the spherical cap of center $\vv$ and angle $\alpha$ as follows:
$$\mathcal{H}_{\vec{v}, \alpha} := \{\vec{x} \in \mathcal{S}^{d-1} ~|~ \theta(\vec{x}, \vec{v}) \leqslant \alpha \}$$
and $\vv$ survives the filter $f_\alpha(\sv)$ iff. $\vv \in \mathcal{H}_{{\sv}, \alpha}$.

\begin{proposition}\cite{MV10}\label{volume_cap}
	For an angle $\alpha \in [0, \pi/2]$ and $\vv \in \mathcal{S}^{d-1}$, the ratio of the volume of a spherical cap $\mathcal{H}_{\vv, \alpha}$ to the volume of the sphere $\mathcal{S}^{d-1}$ is
	$$\V_d(\alpha) := \poly(d) \cdot \sin^d(\alpha).$$
\end{proposition}

\begin{proposition}\cite{BDGL16}
	For an angle $\alpha \in [0,\pi/2]$ and two vectors $\vv, \wv \in \mathcal{S}^{d-1}$ such that $\theta(\vv, \wv) = \theta$, the ratio of the volume of a wedge $\mathcal{H}_{\vv, \alpha} \cap \mathcal{H}_{\wv, \alpha}$ to the volume of the sphere $\mathcal{S}^{d-1}$ is
	
	$$\W_d(\alpha, \theta) := \poly(d) \cdot \Bigg( 1- \frac{2 \cos^2(\alpha)}{1 + \cos(\theta)} \Bigg)^{d/2}.$$
\end{proposition}

\paragraph{Random product codes (RPC).}
LSF method from \cite{BDGL16} uses a decoding oracle that returns relevant filters for a given vector. 
A random-product code (RPC) that admits a fast list-decoding algorithm is sampled over the sphere $\S^{d-1}$. 
The filters are determined by its code words, and the oracle is its decoding algorithm. 

We assume $d=m \cdot b$, for $m=O(\polylog (d))$ and a block size $b$. The vectors in $\mathbb{R}^d$ will be identified with tuples of $m$ vectors in $\mathbb{R}^b$. 
A random product code $C$ of parameters $[d,m,B]$ on subsets of $\mathbb{R}^d$ and of size $B^m$ is defined as a code of the form
$C = Q \cdot (C_1 \times C_2 \times \cdots C_m),$ 
where $Q$ is a uniformly random rotation over $\mathbb{R}^d$ and the subcodes $C_1, ..., C_m$ are sets of $B$ vectors, sampled uniformly and independently random over the sphere $\sqrt{1/m} \cdot \S^{b-1}$, so that codewords are points of the sphere $S^{d-1}$. We can have a full description of $C$ by storing $mB$ points corresponding to the codewords of $C_1,\dots,C_m$ and by storing the rotation $Q$. When the context is clear, $C$ will correspond to the description of the code or to the set of codewords. Random product codes can be easily decoded in some parameter range:

\begin{proposition}[\cite{BDGL16}]\label{Proposition:Decoding}
Let $C$ be a random product code of parameters $[d,m,B]$ with $m = \log(d)$ and $B^m = N^{O(1)}$. For any $\vv \in S^{d-1}$ and $\alpha \in [0,\pi/2]$, one can compute $\mathcal{H}_{\vv,\alpha} \cap C $ in time $N^{o(1)} \scdot |\mathcal{H}_{\vv,\alpha} \cap C|$. 
\end{proposition}

\noindent We can now present the NV-sieve with LSF.

\paragraph{The NV-sieve with LSF.}
Let $\alpha \in [\pi/3, \pi/2]$ be an angle. 
The NV-sieve with LSF \cite{BDGL16} takes as input a list of lattice vectors lying on $\mathcal{S}^{d-1}$ and a constant $\gamma < 1$. 
This algorithm runs in two phases. 
First, during the processing, it samples a random-product code $C$ on the sphere, whose words give the $\alpha$-filters. 
The algorithm decodes half of the vectors of the list to get their nearest $\alpha$-filters, and then add the vectors to the buckets associated to their $\alpha$-filters. 

Secondly, there is the queries phase. For each vector $\vv$ from the other half of the list, the algorithm computes the $\alpha$-filters of $\vv$, and for each vector $\wv$ having a common $\alpha$-filter with $\vv$, the algorithm checks whereas $\|\vv - \wv\| \leqslant \gamma$. 
If it is the case, $\vv - \wv$ is added to the output list. 
This algorithm solves SVP time\footnote{We consider here the values of the NV-sieve, which have asymptotically better space requirements even though this sieve is less efficient in practice.} $2^{0.292d +o(d)}$, and in space $2^{0.208d + o(d)}$ with its space-efficient version \cite{Laa16}. Applying a Grover search instead of testing each candidate in the filter gives the quantum NV-sieve with LSF \cite{Laa16}, which run in same space and in time $2^{0.265d + o(d)}$.

\subsubsection{Probabilistic argument.}

If we consider any vector $\wv \in \S^{d-1}$ and $N^{\rho_0}$ random points $\cv_1,\dots,\cv_{N^{\rho_0}}$ in $\mathcal{H}_{\vv,\alpha}$ for $\rho_0 := \frac{\V_d(\alpha)}{\W_d(\alpha,\theta)}$ ; then this proposition implies that there exists, with constant probability, an $i \in [N^{\rho_0}]$ such that $\cv_i \in \mathcal{H}_{\wv,\alpha}$. 

	Consider a set $S = \sv_1,\dots,\sv_M$ points taken from the uniform distribution on the sphere $S^{d-1}$ and $\vv$ another point randomly chosen on the sphere. Fix also an angle $\alpha \in (0,\pi/2)$. We have the following statements:
	
	\begin{proposition}
		$\forall i \in [M], \ \Pr[\vv \in \hh_{\sv_i,\alpha}] = \Pr[\sv_i \in \hh_{\vv,\alpha}]  = \V_d(\alpha)$. \end{proposition}
	\begin{proof} Immediate by definition of $\V_d(\alpha)$ considering that both $\vv$ and $\sv_i$ are uniform random points on the sphere. \end{proof}
	From the above proposition, we immediately have that $\E[|S \cap \hh_{\vv,\alpha}|] = M\V_d(\alpha)$. We now present a standard concentration bound for this quantity.
\COMMENT{\begin{proposition}\label{Proposition:Chernoff1}
	Assume we have $M\V_d(\alpha)  = 1$. Then 
	$$\Pr[|S \cap \hh_{\vv,\alpha}| \ge d^2 + 1] = e^{-\frac{d^4}{2+d^2}}.$$
\end{proposition}
\begin{proof}
	Let $X_i$ be the random variable which is equal to $1$ if $\sv_i \in \hh_{\vv,\alpha}$ and is equal to $0$ otherwise. We have that each $X_i$ is equal to $1$ with probability $\V_d(\alpha)$. Let $Y = \sum_{i = 1}^M X_i$ so $\E[Y] = 1$. $Y$ is equal to the quantity $|S \cap \hh_{\vv,\alpha}|$.  A direct application of the multiplicative Chernoff bound gives 
	$$ \Pr[Y \ge d^2 + 1] \le e^{-\frac{d^4}{2+d^2}}$$
	which is the desired result. 
\end{proof}}
\begin{proposition}\label{Proposition:Chernoff2}
	Assume we have $M\V_d(\alpha)  = N^x$ with $x > 0$ an absolute constant. Then 
	$$\Pr[|S \cap \hh_{\vv,\alpha}| \ge 2N^x] \le e^{-\frac{N^x}{3}}.$$
\end{proposition}
\begin{proof}
	As before, let $X_i$ be the random variable which is equal to $1$ if $\sv_i \in \hh_{\vv,\alpha}$ and is equal to $0$ otherwise. Let $Y = \sum_{i = 1}^M X_i$ so $\E[Y] = N^x$. $Y$ is equal to the quantity $|S \cap \hh_{\vv,\alpha}|$.  A direct application of the multiplicative Chernoff bound gives 
	$$ \Pr[Y \ge 2N^x] \le e^{-\frac{N^x}{3}}$$
	which is the desired result. 
\end{proof}

%This proposition implies that if we have $N^{\rho_0}$ random points $\cv_1,\dots,\cv_{N^{\rho_0}}$ in $\mathcal{H}_{\vv,\alpha}$ with $\rho_0 = \frac{\V_d(\alpha)}{\W_d(\alpha,\theta)}$, then with constant probability, there will be an $i$ such that $\cv_i \in \mathcal{H}_{\wv,\alpha} $ as well. 

\section{General framework for sieving algorithms using LSF}\label{Section:AlgorithmSkeleton}
We present here a general framework for sieving algorithms using LSF. We present here one sieving step where we start from a list $L$ of $N' = N^{1 + o(1)}$ lattice vectors of norm $1$ and output $N'$ lattice vectors of norm $\gamma < 1$. Sieving algorithms for SVP then consists of applying this subroutine $\poly(d)$ times (where we renormalize the vectors at each step) to find at the end a small vector. We can actually take $\gamma$ very close to $1$ at each iteration, and we refer for example to $\cite{NV08}$ for more details. This framework will encompass the best classical and quantum sieving algorithms.
\begin{algorithm}[H]
	\caption{Sieving algorithms using LSF with parameter $c$} \label{algomain0}
	\textbf{Input:} a list $L$ of $\Nd' = N^{1 + o(1)}$ lattice vectors of norm 1, a constant $\gamma < 1$ and parameter $c \in (0,1)$. \\
	\textbf{Output:} a list $L'$ of $\Nd'$ lattice vectors of norm at most $\gamma$. \\ 
	\textbf{Algorithm:} 
	\begin{algorithmic}
	\STATE $ L' := \{\}$ (empty list) 
	\WHILE{$|L'| \le \Nd'$} 
	\STATE Sample a random product code $C$ of parameter $[d,\log(d),\Nd^{\frac{1-c}{\log(d)}}]$. Let $\cv_1,\dots,\cv_{\Nd^{1-c}}$ be the code points of $C$ and let $\alpha \in [\pi/3,\pi/2] \textrm{ st. } \V_d(\alpha) = \frac{1}{\Nd^{1-c}}$.
	\FOR {$\vv$ in $L$}
	\STATE Add $\vv$ to its $\alpha$-filter's buckets $f_\alpha(\cv_i)$
	\ENDFOR 
	\FOR {\textbf{each} $i \in [\Nd^{1-c}]$}
	\STATE $S \leftarrow {\FINDALLSOLUTIONS}(f_\alpha(\cv_i),\gamma)$
	\STATE $L' := L' \cup S$
	\ENDFOR
	\ENDWHILE
	\STATE \textbf{return} $L'$
	\end{algorithmic}\label{Algorithm:Main}
\end{algorithm}

The $\FINDALLSOLUTIONS(f_\alpha(\sv_i),\gamma)$ subroutine starts from a list of vectors $\xv_1,\dots,\xv_{N^c} \in f_\alpha(\sv_i)$ and outputs all vectors of the form $\xv_i \pm \xv_j$ (with $i \neq j$) of norm less than $\gamma$. We want to find asymptotically all the solutions and not strictly all of them. Let's say here we want to output half of them. Sometimes, there are no solutions so the algorithm outputs an empty list. 

\subsection{Analysis of the above algorithm}
\subsubsection{Heuristics and simplifying assumptions.}

 We first present the heuristic arguments and simplifying assumptions we use for our analysis. 
\begin{enumerate}
	\item The input lattice points behave like random points on the sphere $\S^{d-1}$. The relevance of this heuristic has been studied and confirmed in a few papers starting from the initial NV-sieve \cite{NV08}.
	\item The code points of $C$ behave like random points of the sphere $\S^{d-1}$. This was argued in \cite{BDGL16}, see for instance Lemma 5.1 and Appendix C therein.
	\item We assume that a random point in $f_\alpha(\cv_i)$ is on the border of the filter, {\ie } that it can be written $\xv = \cos(\alpha) \cv_i + \sin(\alpha) \yv$ with $\yv \bot \cv_i$ and of norm $1$. As we argue below, this will be approximately true with very high probability.
\end{enumerate}

In order to argue point $3$, notice that for any angle $\alpha \in (\pi/4, \pi/2)$ and $\eps > 0$, we have $\V_d(\alpha) \gg \V_d(\alpha-\epsilon)$. 
Indeed, for an angle $\epsilon>0$, $\V_d(\alpha-\epsilon)=\sin^d(\alpha-\epsilon) = \V_d(\alpha) \cdot (\epsilon')^d$ with $\epsilon' = \cos\epsilon - \frac{\sin \epsilon \cos \alpha}{\sin \alpha} < 1$ for $\alpha > \epsilon$. 
So the probability for a point to be at angle $\alpha$ with the center of the cap is exponentially higher than to be at angle $\alpha-\epsilon$. That justifies that with very high probability, points in $f_\alpha(\cv_i)$ lie at the border of the cap and hence justifies point $3$.

\subsubsection{Completion.} We start from a list $L$ of $N'$ points. The heuristic states that each point in $L$ is modeled as a random point on the sphere $S^{d-1}$ so each pair of points $\xv,\xv' \in L$ reduces with probability $\V_d(\pi/3) = \frac{1}{N}$. Since there are $\frac{N'(N'-1)}{2}$ pairs of points in $L$, we have on average $\frac{N'(N'-1)}{2N}$ pairs in $L$ that are reducible. We can take for example $N' = 6N$ to ensure that there are on average $\approx 3N'$ pairs.   Therefore, each time we find a random reducible pair, with probability at least $\frac{3N' - |L'|}{3N'} \ge 2/3$, it wasn't already in the list $L'$.

Throughout the rest of the paper

\subsubsection{Time analysis.}
\paragraph{Condition of reduction of vectors.}
Consider two random vectors $\xv_0$ and $\xv_1$ that are in the same $\alpha$-filter's bucket of center $\cv$. We write
\begin{align}
\label{Eq:1} \xv_0 & = \cos(\alpha) \cv + \sin(\alpha) \yv_0 \\ 
\label{Eq:2} \xv_1 & = \cos(\alpha) \cv + \sin(\alpha) \yv_1
\end{align}
with $\yv_0, \yv_1$ of norm 1 both orthogonal to $\cv$. The vectors $\yv_0$ and $\yv_1$ are called residual vectors and if 
$\xv_0$ and $\xv_1$ are random vectors of $f_\alpha(\cv)$ then $\yv_0,\yv_1$ are random vectors in the sphere of dimension $d-2$ of vectors of norm $1$ orthogonal to $\cv$.

\begin{proposition} \label{theta_star}
	Using the notations just above, we have
	$$\theta(\yv_0,\yv_1) \leqslant 2 \arcsin\big( \frac{1}{2 \sin(\alpha)} \big) \Longleftrightarrow \theta(\xv_0,\xv_1) \leqslant \frac{\pi}{3}.$$
\end{proposition}

\begin{proof}
	We denote for simplicity $\theta_y := \theta(\yv_0,\yv_1)$. By subtracting Equation \ref{Eq:2} from Equation \ref{Eq:1} and then by squaring, we have 
	\begin{alignat}{2}
	\|\xv_0 - \xv_1\|^2 \leqslant 1 & \Leftrightarrow \sin^2(\alpha) \|\yv_0-\yv_1\|^2 \leqslant 1 \notag \\
	& \Leftrightarrow \sin^2(\alpha) (2-2 \cos(\theta_y)) \leqslant 1 \notag \\
	& \Leftrightarrow \cos(\theta_y) \geqslant 1 - \frac{1}{2 \sin^2(\alpha)} \notag \\
	& \Leftrightarrow \theta_y \leqslant \arccos \big( 1-\frac{1}{2 \sin^2(\alpha)}\big) = 2 \arcsin \big( \frac{1}{2 \sin(\alpha)} \big) \text{, true for } \alpha \leqslant \pi/2. \notag
	\end{alignat}
\end{proof}

\begin{corollary}\label{delta}
	Let $\alpha \in [\pi/3, \pi/2]$ and a random pair of vectors in a same $\alpha$-filter's bucket. 
	The probability that the pair is reducing is equal to $\V_{d-1}(\theta^*_\alpha)$, with 
	$$\theta^*_\alpha := 2 \arcsin \Big( \frac{1}{2 \sin(\alpha)} \Big).$$
\end{corollary}

Notice that we have $\V_{d-1}$ because we work with residual vectors (orthogonal to $\sv$) but since $\V_d$ and $\V_{d-1}$ are asymptotically equivalent, we will keep writing $\V_d(\theta^*_\alpha)$ everywhere for simplicity.
From the above corollary, we have that for an $\alpha$-filter that has $\Nd^c$ points randomly distributed in this filter, the expected number of reducing pairs is $\Nd^{2c} \cdot \V_{d-1}(\theta^*_\alpha)$.

\begin{proposition}\label{Proposition:AlgorithmGeneral}
	Consider Algorithm \ref{Algorithm:Main} with parameter $c \in [0,1]$ and associated angle $\alpha \in [\pi/3, \pi/2]$ satisfying $\V_d(\alpha) = \Nd^{-(1-c)}$.
	Let $\zeta$ such that $\Nd^{\zeta} = \Nd^{2c} \scdot \V_{d-1}(\theta^*_\alpha)$. The above algorithm runs in time $T = \NBREP \scdot (\INIT + \FAS)$ where 
	$$ {\NBREP} = \max\{1,\Nd^{c - \zeta + o(1)}\} \quad ; \quad \INIT = \Nd^{1 + o(1)} \quad ; \quad \FAS = \Nd^{1-c} \FAS_1$$
	where $\FAS_1$ the running time of a single call to the ${\FINDALLSOLUTIONS}$ subroutine. 
\end{proposition}
\begin{proof}
	We first analyze the two \textbf{for} loops. $\INIT$ is the running time of the first loop. For each point $\vv \in L$, we need to compute $\hh_{\vv,\alpha} \cap C$ and update the corresponding buckets $f_\alpha(\sv_i)$. We have $|C| = N^{1-c}$ and we chose $\alpha$ such that $\V_d(\alpha) = N^{-(1-c)}$, so the expected value of $|\hh_{\vv,\alpha} \cap C|$ is $1$. For each point $\vv$, we can compute $\hh_{\vv,\alpha} \cap C$ in time $N^{o(1)}|\hh_{\vv,\alpha}|$ using Proposition \ref{Proposition:Decoding}. From there, we can conclude that we compute the filter for the $N'$ points in time $\INIT = N^{1 + o(1)}$.
	
	The second loop runs in time $\FAS = \Nd^{1-c} \FAS_1$ by definition. After this loop, the average number of solutions found is $\Nd^{\zeta}$ for each call to $\FINDALLSOLUTIONS$ so $\Nd^{1-c+\zeta} $ in total (notice that we can have $\zeta < 0$, which means that we can find on average much less that one solution for each call of $\FINDALLSOLUTIONS$). We run the \textbf{while} loop until we find $\Nd'$ solutions so we must repeat this process $\NBREP = \max\{1,\Nd^{1 - (1 - c + \zeta) + o(1)}\} = \max\{1,\Nd^{c - \zeta + o(1)}\}$ times.
\end{proof}

This formulation of sieving algorithms is easy to analyze. Notice that the above running time depends only on $c$ (since $\alpha$ can be derived from $c$ and $\zeta$ can be derived from $c,\alpha$) and on the $\FINDALLSOLUTIONS$ subroutine. We now retrieve the best known classical and quantum sieving algorithm in this framework. \\ \\
\emph{Best classical algorithm.} In order to retrieve the time exponent of \cite{BDGL16}, we take $c \rightarrow 0$, which implies $\alpha \rightarrow \pi/3$. We can compute $\theta^*_{\pi/3} \approx 1.23\text{rad} \approx 70.53\degree$ and $\zeta = -0.4094$. In this case, we have $\FAS_1 = O(1)$. From the above proposition, we get a total running time of $T = \Nd^{1.4094 + o(1)} = 2^{0.2925d + o(d)}$. \\ \\
\emph{Best quantum algorithm.} In order to retrieve the time exponent of \cite{Laa16}, we take $c = 0.2782$. This value actually corresponds to the case where $\zeta = 0$, so we have on average one solution per $\alpha$-filter. For the $\FINDALLSOLUTIONS$ subroutine, we can apply Grover's algorithm on pairs of vectors in the filter to find this solution in time $\sqrt{\Nd^{2c}} = \Nd^{c}$ (there are $\Nd^{2c}$ pairs) so $\FAS_1 = \Nd^c$. Putting this together, we obtain $T = \Nd^{1+c + o(1)} = \Nd^{1.2782 + o(1)} = 2^{0.2653d + o(d)}$. \\

In the next section, we show how to improve the above quantum algorithm. Our main idea is to replace Grover's algorithm used in the $\FINDALLSOLUTIONS$ subroutine with a quantum random walk. In the next section, we present the most natural quantum walk which is done over a Johnson graph and where a vertex is marked if the points of a vertex contain a reducible pair, in a similar way than for element distinctness. We then show in a later section how this random walk can be improved by relaxing the condition on marked vertices. 
\section{Quantum random walk for the $\FINDALLSOLUTIONS$ subroutine: a first attempt}\label{Section:QuantumWalk}
\subsection{Constructing the graph}\label{Section:QuantumWalkDescription}
We start from an unordered list $\xv_1,\dots,\xv_{N^c}$ of distinct points in a filter $f_\alpha(\sv)$ with $\alpha$ satisfying $\V_d(\alpha) = \frac{1}{N^{1-c}}$. Let $L_x$ be this list of $\xv_i$. For each $i \in [N^c]$, we write $\xv_i = \cos(\alpha) \sv + \sin(\alpha)\yv_i$ where each $\yv_i$ is of norm $1$ and orthogonal to $\sv$. Recall from Proposition \ref{theta_star} that a pair $(\xv_i,\xv_j)$ is reducible iff. $\theta(\yv_i,\yv_j) = \theta^*_\alpha = 2\arcsin(\frac{1}{2\sin(\alpha)})$. We will work only on the residual vectors $\yv_i$ and present the quantum random walk that finds pairs $\yv_i,\yv_j$ such that $\theta(\yv_i,\yv_j) = \theta^*_\alpha$ more efficiently than with Grover's algorithm. Let $L_y = \yv_1,\dots,\yv_{N^c}$ be the list of all residual vectors. 

The quantum walk has two extra parameters $c_1 \in [0,c]$ and $c_2 \in [0,c_1]$. From these two parameters, let $\beta \in [\pi/3,\pi/2]$ \text{st.} $\V_{d}(\beta) = N^{c_2 - c_1}$ and ${\rho_0}$ \textrm{st.} $N^{\rho_0} = \frac{\V_d(\beta)}{\W_d(\beta,\theta^*_\alpha)}$. We start by sampling a random product code $\C_2$ with parameters $[(d-1),\log(d-1),N^{\frac{{\rho_0} + c_1 - c_2}{\log(d-1)}}]$ which has therefore $N^{{\rho_0} + c_1 - c_2} = \frac{1}{\W_d(\beta,\theta^*_\alpha)}$ points denoted $\tv_1,\dots,\tv_{N^{{\rho_0} + c_1 - c_2}}$. We perform our quantum random walk on a graph $G = (V,E)$ where each vertex $v \in V$ contains:
\begin{itemize}
	\item An unordered list $L^v_y = \yv_1,\dots,\yv_{N^{c_1}}$ of distinct points taken from $L_y$.
	\item For each $\tv_i \in \C_2$, we store the list of elements of $ J^v(\tv_i) := f_{\beta}(\tv_i) \cap L^v_y$. For each $\tv_i$, we do this using a quantum data structure that stores $J^v(\tv_i) $ where we can add and delete efficiently in quantum superposition. This can be done with QRAM. Notice that we have on average 
	$$|J^v(\tv_i) |= N^{c_1} \scdot \V_{d}(\beta) = N^{c_2},$$
 	and we need to store in total $|\C_2| \scdot N^{c_2} = N^{c_1 + {\rho_0}}$ such elements in total for each vertex.
	\item A bit that says whether the vertex is marked (we detail the marked condition below). 
\end{itemize}

The vertices of $G$ consists of the above vertices for all possible lists $L^v_y$. We have $(v,w) \in E$ if we can go from $L_y^v$ to $L_y^w$ by changing exactly one value. In order words 
$$ (v,w) \in E \Leftrightarrow \exists \yv_{old} \in L_y^v \text{ and } \yv_{new} \in L_y \backslash L_y^v \text{ st. } L_y^w = \left(L_y^v\backslash\{\yv_{old}\}\right)\cup\{\yv_{new}\}.$$

This means the graph $G$ is exactly a Johnson graph $J(N^c,N^{c_1})$ where each vertex also has some additional information as we described above. Once we find a marked vertex, it contains a pair $(\yv_i,\yv_j)$ such that $\theta(\yv_i,\yv_j) \le \theta^*_\alpha$ from which we directly get a reducible pair $(\xv_i,\xv_j)$.

\paragraph{Condition for a vertex to be marked.}
We define the following subsets of vertices . We first define the set $M_0$ vertices for which there exists a pair of points which is reducible.
$$M_0 \eqdef \{v \in V : \exists \yv_i, \yv_j \neq \yv_i \in L^v_y, \theta(\yv_i,\yv_j) \le \theta^*_\alpha\}.$$
Ideally, we would want to mark each vertex in $M_0$, however this would induce a too large update cost when updating the bit that specifies whether the vertex is marked or not. Instead, we will consider as marked vertices subsets of $M_0$ but for which the update can be done more efficiently, but losing only a small fraction of the marked vertices. For each $J^v(\tv_i)$, we define $\widetilde{J}^v(\tv_i)$ which consists of the first $2N^{c_2}$ elements of $J^v(\tv_i)$\footnote{We consider an global ordering of elements of $L_y$, for example with respect to their index, and $J^v(\tv_i)$ consists of  the $2N^{c_2}$ elements of ${J}^v(\tv_i)$ which are the smallest with respect to this ordering.} and if  $|J^v(\tv_i)| \le 2N^{c_2}$, we have $\widetilde{J}^v(\tv_i) = J^v(\tv_i)$ . We define the set of marked elements $M$ as follows:
$$M \eqdef \{v \in V : \exists \tv \in \C_2, \exists \yv_i, \yv_j \neq \yv_i \in \widetilde{J}^v(\tv), \textrm{ st. } \theta(\yv_i,\yv_j) \le \theta^*_\alpha\}.$$

The reason for using such a condition for marked vertices is that when we will perform an update, hence removing a point $\yv_{old}$ from a vertex and adding a point $\yv_{new}$, we will just need to look at the points in $\widetilde{J}^v(\tv)$ for $\tv \in  f_\beta(\yv_{new}) \cap \C_2$ which can be done faster than by looking at all the points of the vertex. If we used $J^v(\tv)$ instead of $\widetilde{J}^v(\tv)$ then the argument would be simpler but we would only be able to argue about the average running time of the update but the quantum walk framework require to bound the update for any pair of adjacent vertices\footnote{This problem arises in several quantum random walk algorithms, for example for quantum subset-sum algorithms. One solution is to use a heuristic that essentially claims that we can use the average running time of the update cost instead of the worst case. In our case, we don't need this heuristic as we manage to bound the update cost in the worst case.  We refer to \cite{BBSS20} for an interesting discussion on the topic.}. Also notice that each vertex still contain the sets $J^v(\tv_i)$ (from which one can easily compute $\widetilde{J}^v(\tv_i)$).

\subsection{Time analysis of the quantum random walk on this graph}\label{Section:QwalkTimeAnalysis}
We are now ready to analyze our quantum random walk, and compute its different parameters. Throughout our analysis, we define $K(\yv_i) \eqdef  f_\beta(\yv_i) \cap \C_2$ and we have on average
$$|K(\yv_i)| = N^{{\rho_0} + c_1 - c_2} \cdot \V_{d}(\beta) = N^{{\rho_0}}.$$
Using Proposition \ref{Proposition:Chernoff2}, we have for each $i$, 
\begin{align}
	\Pr[|K(\yv_i)|]>  2N^{\rho}]  \le  e^{-\frac{N^{\rho_0}}{3}}
\end{align}
and using a union bound, we have for any absolute constant $\rho_0 > 0$:
\begin{align}
	\Pr[\forall i \in [N^c], \ |K(\yv_i)| \le  2N^{\rho}] \ge 1 - N^ce^{-\frac{N^{\rho_0}}{3}} = 1 - o(1).
\end{align}
So for a fixed $\alpha$-filter, we have with high probability that each $|K(\yv_i)|$ is bounded by $2N^{\rho_0}$ and we assume we are in this case. The sets $K(\yv_i)$ can hence be constructed in time $N^{{\rho_0} + o(1)}$ using the decoding procedure (Proposition \ref{Proposition:Decoding}) for $\C_2$.

\paragraph{Setup cost.} In order to construct a full vertex $v$ from a list $L^v_y = \yv_1,\dots,\yv_{N^{c_1}}$, the main cost is to construct the lists $J^v(\tv_i) = f_\beta(\tv_i) \cap L^v_y$. To do this, we start from empty lists $J^v(\tv_i)$. For each $\yv_i \in L^v_y$, we construct the list $K(\yv_i) = f_\beta(\yv_i) \cap \C_2$ and for each codeword $\tv_j \in K(\yv_i)$, we add $\yv_i$ in $J^v(\tv_i)$. 

This takes time $N^{c_1} \scdot N^{{\rho_0} + o(1)}$. We can perform a uniform superposition of the vertices by performing the above procedure in quantum superposition. This can also be done in $N^{c_1} \scdot N^{{\rho_0} + o(1)}$ since we use a quantum data structure that performs these insertions in $J^v(\tv_i)$ efficiently. So in conclusion, 

$$ \S = N^{c_1 + {\rho_0} + o(1)}.$$

\paragraph{Update cost.} We show here how to go from a vertex $v$ with associated list $L_y^v$ to a vertex $w$ with $L_y^w = \left(L_y^v\backslash \{\yv_{old}\}\right) \cup \{\yv_{new}\}$. We start from a vertex $v$ so we also have the lists $J^v(\tv_i)= f_\beta(\tv_i) \cap L^v_y$. 

In order to construct the lists $J^w(\tv_i)$, we first construct $K(\yv_{old}) = f_\beta(\yv_{old}) \cap \C_2$ and for each $\tv_i$ in this set, we remove $\yv_{old}$ from $J^v(\tv_i)$. Then, we construct $K(\yv_{new}) $ and for each $\tv_i$ in this set, we add $\yv_{new}$ to $J^v(\tv_i)$, thus obtaining all the $J^w(\tv_i)$. Constructing the two lists takes time on average $N^{{\rho_0} + o(1)}$ and we then perform at most $2N^{{\rho_0}}$ deletion and insertion operations which are done efficiently. These operations take $N^{{\rho_0} + o(1)}$ deletions and insertions, which can be done efficiently. 

If $v$ was marked and $\yv_{old}$ is not part of the reducible pair then we do not change the last registers for $L_y^w$. If $v$ was not marked, then we have to ensure that adding $\yv_{new}$ doesn't make it marked. So we need to check whether there exists $\vec{y'} \neq\yv_{new}$ such that 
$$ \exists \tv \in \C_2, \yv_{new},\vec{y_0} \in \widetilde{J}^w(\tv) \text{ and } (\yv_{new},\yv_0) \text{ are reducible }.$$ 
 If such a point $\yv_0$ exists, it necessarily lies in the set $\cup_{\tv \in K(\yv_{new})} \widetilde{J}^v(\tv)$ which is of size at most $2N^{\rho}\scdot 2N^{c_2} = 4N^{\rho_0 + c_2}$. We perform a Grover search on this set to determine whether there exists a $\yv_0 \in \cup_{\tv \in \C_2} \widetilde{J}^v(\tv)$ that reduces with $\yv_{new}$, and this takes time $N^{\frac{{\rho_0} + c_1 + o(1)}{2}}$. In conclusion, we have that the average update time is 
 $$ \U = N^{{\rho_0} + o(1)} + N^{\frac{{\rho_0} + c_2 + o(1)}{2}} \le  N^{\max\{{\rho_0},\frac{{\rho_0} + c_2}{2}\} + o(1)}.$$

 \paragraph{Checking cost.} Each vertex has a bit that says whether it is marked or not so we have 
 $$\C = 1.$$
 
 \paragraph{Computing the fraction of marked vertices Epsilon.} 
 We prove here the following proposition
 \begin{proposition}
 	$\eps \ge \Theta\left(\min\left\{N^{2c_1}\V_d(\beta),1\right\}\right).$
 \end{proposition}
\begin{proof}
 We consider a random vertex in the graph and lower bound the probability that it is marked. A sufficient condition for a vertex $v$ to be marked is if it satisfies the following $2$ events :
 \begin{itemize}
 	\item $E_1:$ \ $\exists \tv \in \C_2, \exists \yv_i,\yv_j \neq \yv_i \in J^v(\tv), \textrm{ st. } \theta(\yv_i,\yv_j) \le \theta^*_{\alpha}$.
 	\item $E_2 :$ \ $\forall \tv \in \C_2, |J^v(\tv)|  \le 2N^{c_2}$.
 \end{itemize}
The second property implies that $\forall \tv \in \C_2, \ J^v(\tv) = \widetilde{J}^v(\tv)$ and in that case, the first property implies that $v$ is marked. We now bound the probability of each event

\begin{lemma}\label{Lemma:Lemma1}
	$\Pr[E_1]  \ge \Theta\left(\min\left\{N^{2c_1}\V_d(\beta),1\right\}\right)$.
\end{lemma}

\begin{proof}
	For a fixed pair $\yv_i,\yv_j \neq \yv_i \in L^v_y$, we have $\Pr[\theta(\yv_i,\yv_j) \le \theta^*_\alpha] = \V_d(\theta^*_\alpha)$. Since there are $\Theta(N^{2c_1})$ such pairs, if we define the event $E_0$ as: $\exists \yv_i,\yv_j \neq \yv_i \in L^v_y, \textrm{ st. } \theta(\yv_i,\yv_j) \le \theta^*_\alpha$, we have  
	$$ \Pr[E_0] \ge \Theta\left(\min\left\{N^{2c_1}\V_d(\beta),1\right\}\right).$$
	Now we assume $E_0$ holds and we try to compute the probability that $E_1$ is true conditioned on $E_0$. So we assume $E_0$ and let $\yv_i,\yv_j \neq \yv_i \in L^v_y, \textrm{ st. } \theta(\yv_i,\yv_j) \le \theta^*_\alpha$. For each code point $\tv \in \C_2$, we have 
	\begin{align*}
		\Pr[\yv_i,\yv_j \in J^v(\tv)] = \Pr[\tv \in \hh_{\yv_i,\beta} \cap \hh_{\yv_j,\beta}] = \W_d(\beta,\theta^*_\alpha).
	\end{align*}
  Therefore, we have 
  \begin{align}\label{Eq:OfLemma1}\Pr[\exists \tv \in \C_2, \ \yv_i,\yv_j \in J^v(\tv)] = 1 - (1-\W_d(\beta,\theta^*_\alpha))^{|C_2|}.
  \end{align}
  Since $|\C_2 |= \frac{1}{\W_d(\beta,\theta^*_\alpha)}$, we can conclude 
  $$\Pr[E_1|E_0] \ge \Pr[\exists \tv \in \C_2, \ \yv_i,\yv_j \in J^v(\tv)] = 1 - (1-\W_d(\beta,\theta^*_\alpha))^{|C_2|} \ge \Theta(1),$$
  which implies $\Pr[E_1] \ge  \Pr[E_1|E_0]\scdot \Pr[E_0] \ge \Theta\left(\max\left\{N^{2c_1}\V_d(\beta),1\right\}\right)$.
\end{proof}
\begin{lemma}
	$\Pr[E_2] \ge 1 - |C_2|e^{-\frac{N^{c_2}}{3}}$.
\end{lemma}
\begin{proof}
	For each $\tv \in \C_2$, we have using Proposition \ref{Proposition:Chernoff2} that $\Pr[|J^v(\tv)| \le 2N^{c_2}] \ge 1 - e^{-\frac{N^{c_2}}{3}}$. Using a union bound, we have 
	$$ \Pr[\forall \tv \in \C_2, |J^v(\tv)| \le 2N^{c_2}] \ge 1 - |C_2|e^{-\frac{N^{c_2}}{3}}.$$
\end{proof}
We can now finish the proof of our Proposition. We have 
\begin{align*}
	\eps \ge \Pr[E_1 \wedge E_2] & \ge \Pr[E_1] + \Pr[E_2] - 1 \\ & \ge \Theta\left(\max\left\{N^{2c_1}\V_d(\beta),1\right\}\right) - |C_2|e^{-\frac{N^{c_2}}{3}} \\ & \ge \Theta\left(\max\left\{N^{2c_1}\V_d(\beta),1\right\}\right) 
\end{align*}
The last inequality comes from the fact that $|C_2|e^{-\frac{N^{c_2}}{3}}$ is vanishing doubly exponentially in $d$ ($N$ is exponential in $d$) so it is negligible compared to the first term and is absorbed by the $\Theta(\scdot)$.
\end{proof}

 \paragraph{Computing the spectral gap Delta.} We are in a $J(N^c,N^{c_1})$ Johnson graph so we have 
 $$\delta \approx N^{-c_1}.$$
 
 \paragraph{Running time of the quantum walk.}
 The running time $T_1$ of the quantum walk is (omitting the $o(1)$ terms and the $O(\scdot)$ notations)
\begin{align*}
	T_1 & = \S + \frac{1}{\sqrt{\eps}}\left(\frac{1}{\sqrt{\delta}} ~ \U + \C\right) \\
	& = N^{c_1 + {\rho_0}} + \frac{1}{\max\{1,N^{c_1}\sqrt{\V_d(\theta^*_\alpha)}\}}\left(N^{\max\{{\rho_0},\frac{{\rho_0} + c_2}{2}\} + \frac{c_1}{2}}\right)
\end{align*}
In this running time, we can find one marked vertex with high probability if it exists. We repeat this quantum random walk until we find $\max\{\frac{N^{\zeta}}{2},1\}$ solutions.\\

\begin{minipage}{0.85\textwidth}
	\cadre{
		\begin{center}\textbf{Algorithm for the $\FINDALLSOLUTIONS$ procedure} \end{center}
		Pick a random product code $\C_2$. \\
		\textbf{while} the number of solutions found is $< \frac{N^{\zeta}}{2}$: \\
		$\quad$ Run our QRW to find a solution and add it to the list of solutions if it hasn't been found.
	} 
\end{minipage} $ \ $ \\

For $\zeta > 0$, there are $N^{\zeta}$ different solutions that can be found in each $\alpha$-filter. Each time we find a solution, since the list of solutions found is $< \frac{N^{\zeta}}{2}$. Therefore, the probability that each solutions found by the QRW is new is at least $\frac{1}{2}$. We have therefore 

$$\FAS_1 = \max\{N^{\zeta},1\}\cdot T_1.$$

If $\zeta > 0$, our algorithm finds $\Theta(N^{\zeta})$ solutions in time $N^{\zeta}T_1$ and if $\zeta \le 0$, our algorithm finds $1$ solution in time $T_1$ with probability $\Theta(N^{-\zeta})$.

\subsection{Memory analysis}

\paragraph{Classical space.}
We have to store at the same time in classic memory the $N$ list vectors of size $d$, and the buckets of the $\alpha$-filters. Each vector is in $N^{o(1)}$ $\alpha$-filter, so our algorithm takes classical space $N^{1 + o(1)}$.

\paragraph{Memory requirements of the quantum random walk.}
Each vertex $v$ of the graph stores all the $J^v(\tv_i)$ which together take space $N^{c_1 + {\rho_0}}$. We need to store a superposition of vertices so we need $N^{c_1 + {\rho_0}}$ quantum registers and we need that same amount of QRAM because we perform insertions and deletions in the database in quantum superposition. All the operations require QRAM access to the whole list $L_y$ which is classically stored and is of size $N^c$. Therefore, we also require $N^c$ QRAM.

\subsection{Optimal parameters for this quantum random walk}\label{Section:Optimal1}
Our algorithm takes in argument three parameters: $c \in [0,1]$, $c_1 \le c$ and $c_2 \le c_1$ from which we can express all the other variables we use: $\alpha$, $\theta^*_\alpha$, $\beta$, ${\rho_0}$ and $\zeta$. We recall these expressions as they are scattered throughout the previous sections:
\begin{itemize}
	\item $\alpha$: angle in $[\pi/3,\pi/2]$ that satisfies $\V_d(\alpha) = \frac{1}{N^{1-c}}$.
	\item $\theta^*_\alpha = 2\arcsin(\frac{1}{2\sin(\alpha)})$.
	\item $\beta$: angle in $[\pi/3,\pi/2]$ that satisfies $\V_{d}(\beta) = \frac{1}{N^{c_1 - c_2}}$.
	\item ${\rho_0}$: non-negative real number such that $N^{\rho_0} = \frac{\V_d(\beta)}{\W_d(\beta,\theta^*_\alpha)}$.
	\item $\zeta$: real number such that $N^{\zeta} = N^{2c}\V_{d}(\theta^*_\alpha)$.
\end{itemize}

Plugging the value of $\FAS_1$ from the end of Section \ref{Section:QwalkTimeAnalysis} in Proposition \ref{Proposition:AlgorithmGeneral}, we find that the total running time of our quantum sieving algorithm with parameters $c,c_1,c_2$ is

$$ T = N^{c - \zeta}\left(N + N^{1-c}\max\{N^{\zeta},1\}\left(N^{c_1 + {\rho_0}} + \frac{1}{\max\{1,N^{c_1}\sqrt{\V_d(\theta^*_\alpha)}\}}\left(N^{\max\{{\rho_0},\frac{{\rho_0} + c_2}{2}\} + \frac{c_1}{2}}\right)\right)\right).$$

We ran a numerical optimization over $c,c_1,c_2$ to get our optimal running time, summed up in the following theorem. 

\begin{proposition}\label{previous}
	Our algorithm with parameters 
	$$c \approx 0.3300 \quad ; \quad c_1 \approx 0.1952 \quad ; \quad c_2 \approx 0.0603$$
	heuristically solves SVP on dimension $d$ in time $T = N^{1.2555 + o(1)} = 2^{0.2605d + o(d)}$, uses QRAMM of maximum size $N^{0.3300 + o(1)} = 2^{0.0685d + o(d)}$, a quantum memory of size $N^{0.2555 + o(1)} = 2^{0.0530d + o(d)}$ and uses a classical memory of size $N^{1 + o(1)} = 2^{0.2075d + o(d)}$.
\end{proposition}

With these parameters, we obtain the values of the other parameters:
	$$ \alpha \approx 1.1388\text{rad} \approx 65.25\degree; \quad \theta^*_\alpha \approx 1.1661\text{rad} \approx 66.46\degree; \quad \beta \approx 1.3745\text{rad} \approx 78.75\degree $$ $${\rho_0} \approx 0.0603 ; \quad \zeta \approx 0.0745.$$
As well as the quantum walk parameters:
$$
\S = N^{c_1 + {\rho_0}} = N^{0.2555}; \quad \U = N^{{\rho_0}} = N^{0.0603}; \quad \C = 0; \quad \eps = \delta = N^{-c_1} = N^{-0.1952}.$$

The equality ${\rho_0} = c_2$ allows to balance the time of the two operations during the update step.
With these parameters we also obtain $ \S = \U/{\sqrt{\epsilon ~ \delta}} = N^{c_1 + {\rho_0}} = N^{0.2555d}$, which balances the overall time complexity. 

Notice that with these parameters, we can rewrite $T$ as
$$ T = N^{c - \zeta}\left(N + N^{1 - c + \zeta + c_1 + {\rho_0}}\right) = N^{1 + c - \zeta} + N^{1 + c_1 + {\rho_0}}.$$
Also, we have $c_1 + {\rho_0} = c - \zeta$, which equalizes the random walk step with the initialization step. From our previous analysis, the amount of required QRAM is $N^{c}$ and the amount of quantum memory needed is $N^{c_1 + {\rho_0}}$.

\section{Quantum random walk for the $\FINDALLSOLUTIONS$ subroutine: an improved quantum random walk}\label{Section:QuantumWalk2}

We now add a variable $\rho \in (0,\rho_0]$ that will replace the choice of $\rho_0$ above. $\rho_0$ was chosen in order to make sure that if a pair $\yv_i,\yv_j$ exists in a vertex $v$, then it will appear on one of the $J^v(\tv)$ for $\tv \in \C_2$. However, we can relax this and only mark a small fraction of these vertices. This will reduce the fraction of marked vertices, which makes it harder to find a solution, but having a smaller $\rho$ will reduce the running time of our quantum random walk. 

The construction is exactly the same as in the previous section just that we replace $\rho_0$ with $\rho$. This implies that $|C_2| = N^{\rho + c_1 - c_2}$. We can perform the same analysis as above 

\paragraph{Time analysis of this QRW in the regime $\zeta + \rho - \rho_0 > 0$.}
We consider the regime where $\zeta + \rho - \rho_0 > 0$ and $\rho \in (0,\rho_0]$ (in particular $\zeta > 0$, since $\rho_0 > 0$). This regime ensures that even when if we have less marked vertices, then there on average more than one marked vertex, so our algorithm at least finds one solution with a constant probability.

The analysis walk is exactly the same than in Section \ref{Section:QwalkTimeAnalysis}, each repetition of the quantum random walk takes time $T_1$ with 
$$
T_1 = \S + \frac{1}{\sqrt{\eps}}\left(\frac{1}{\sqrt{\delta}} ~ \U + \C\right) $$
with 
\begin{align*} \S & = N^{c_1 + \rho}, \quad \U = N^{\max\{\rho,\frac{\rho + c_2}{2}\} + o(1)}, \quad \C = 1, \\
	\eps & = N^{2c_1} N^{\rho - \rho_0}\V_d(\theta^*_\alpha), \quad \delta = N^{-c1}.\end{align*}

The only thing maybe to develop is the computation of $\eps$. We perform the same analysis as above but with $|C_2| = N^{\rho + c_1 - c_2}$. This means that Equation \ref{Eq:OfLemma1} of Lemma \ref{Lemma:Lemma1} becomes 
  \begin{align*}\Pr[\exists \tv \in \C_2, \ \yv_i,\yv_j \in J^v(\tv)] & = 1 - (1-\W_d(\beta,\theta^*_\alpha))^{|C_2|} \\ &\ge |C_2|\W_d(\beta,\theta^*_\alpha) = N^{\rho - \rho_0}.
\end{align*}
which gives the extra term $N^{\rho - \rho_0}$ in $\eps$. Another issue is that now, we can only extract $N^{\zeta + \rho - \rho_0}$ solutions each time we construct the graph, we have therefore to repeat this procedure to find $\frac{N^{\zeta + \rho - \rho_0}}{2}$ solutions with this graph and then repeat the procedure with a new code $\C_2$. The algorithm becomes \\

\begin{minipage}{0.85\textwidth}
	\cadre{
		\begin{center}\textbf{Algorithm from Section \ref{Section:QuantumWalk2} with parameter $\rho$} \end{center}
		\textbf{while} the total number of solutions found is $< \frac{N^{\zeta}}{2}$: \\
		$\quad$ Pick a random product code $\C_2$. \\
		$\quad \ $\textbf{while} the number of solutions found is $< \frac{N^{\zeta + \rho - \rho_0}}{2}$ with this $\C_2$: \\
		$\quad \quad $ Run our QRW with $\rho$ to find a new solution.
	} 
\end{minipage} \\

With this procedure, we also find $\Theta(N^{\zeta})$ solutions in time $N^{\zeta}T_1$ and $\FAS_1 = N^{\zeta}T_1$ (Recall that we are in the case $\zeta \ge \zeta + \rho - \rho_0 > 0$). Actually, Optimal parameters will be when $c_2 = 0$ and $\rho \rightarrow 0$.

% 1/ Calculer tous les filtres et insérer dans $N^\rho$ au hasard => Pire que l'algo précédent, même setup et update, pire 1/sqrt(eps).

%2/ Calculer les $N^\rho$ premiers filtres => Est-ce que qu'ils se comportent bien comme des points aléatoires ? 

%3/ Modifier l'algo de décodage pour que les $N^\rho$ filtres les plus proches soient randomisés. 

\subsection{Analysis of the above algorithm}
This change implies that some reducing pairs are missed. For the quantum random walk complexity, this only change the probability, denoted $\epsilon$, so that a vertex is marked. 
Indeed, it is equal to the one so that there happens a collision between two vectors through a filter, which is no longer equal to the existence of a reducing pair within the vertex. Indeed, to have a collision, there is the supplementary condition of both vectors of a reducing pair are inserted in the same filter, which is of probability $N^{\rho_0-\rho}$. So we get a higher value of $\epsilon = N^{2c_1} \V_d(\theta_\alpha^*) \cdot N^{\rho_0-\rho}$. 

However, this increasing is compensated by the reducing of the costs of the setup ($N^{c_1+\rho+o(1)}$) and the update ($2 N^{\max\{\rho, \frac{\rho+c_2}{2}\}+o(1)}$). 

A numerical optimisation over $\rho, c, c_1$ and $c_2$ leads to the following theorem. 

\begin{theorem}[Theorem \ref{Theorem:Main} restated]
	Our algorithm with a free $\rho$ with parameters 
	$$\rho \rightarrow 0 \quad ; \quad c \approx 0.3696 \quad ; \quad c_1 \approx 0.2384 \quad ; \quad c_2 = 0$$
	heuristically solves SVP on dimension $d$ in 
	time $T = N^{1.2384 + o(1)} = 2^{0.2570 d + o(d)}$, 
	uses QRAM of maximum size $N^{0.3696} = 2^{0.0767d}$,
	a quantum memory of size $N^{0.2384} = 2^{0.0495d}$ and uses a classical memory of size $N^{1 + o(1)} = 2^{0.2075d + o(d)}$.
\end{theorem}

 With these parameters, we obtain the values of the other parameters:
$$\alpha \approx 1.1514 \text{ rad}; \quad \theta^*_\alpha \approx 1.1586 \text{ rad}; \quad \beta \approx 1.1112 \text{ rad}; \quad \zeta \approx 0.1313.$$
As well as the quantum walk parameters:
$$
\S = N^{c_1 + \rho} = N^{0.2384}; \quad \U = N^{\rho} = N^{o(1)}; \quad \C = 0; \quad \eps = \delta = N^{-c_1} = N^{-0.2384}.$$ 

With these parameters, we also have $\rho_0 = 0.107$ so we are in the regime where $\zeta + \rho - \rho_0 > 0$. As in the previous time complexity stated in Theorem \ref{previous}, we reach the equality $S = U/\sqrt{\epsilon \delta}$, which allows to balance the time of the two steps of the quantum random walk: the setup and the search itself. 

Notice that with these parameters, we can rewrite $T$ as 
$$ T = N^{c - \zeta}\left(N + N^{1 - c + \zeta + c_1 + \rho}\right) = N^{1 + c - \zeta} + N^{1 + c_1 + \rho}.$$ 
With our optimal parameters, we have $\rho = 0$ and $c - \zeta = c_1$, which equalizes the random walk step with the initialization step. 
From our previous analysis, the amount of required QRAM is $N^{c}$ and the amount of quantum memory needed is $N^{c_1}$. 

\section{Space-time trade-offs} 

By varying the values $c, c_1, c_2$ and $\rho$, we can obtain trade-offs between QRAM and time, and between quantum memory and time. 
All the following results come from numerical observations.

\subsection{Trade-off for fixed quantum memory.}

We computed the minimized time if we add the constraint that the quantum memory must not exceed $2^{Md}$. 
For a chosen fixed $M$, the quantum memory is denoted is $2^{\mu_M d} = 2^{Md}$ and the corresponding minimal time by $2^{\tau_M d}$. 
The variation of $M$ also impacts the required QRAM to run the algorithm, that we denote by $2^{\gamma_M d}$.

So we get a trade-off between time and quantum memory in Figure \ref{fig:3time_qspace}, and the evolution of QRAM in function of $M$ for a minimal time is in Figure \ref{fig:3QRACM_qspace}.

\begin{figure}[H]
\centering
\includegraphics[scale = 0.6]{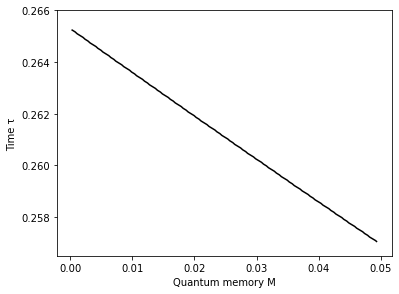}
\caption{Quantum memory-time trade-off.} 
\label{fig:3time_qspace} 
\end{figure} 

\begin{figure}[H]
\centering
\includegraphics[scale = 0.6]{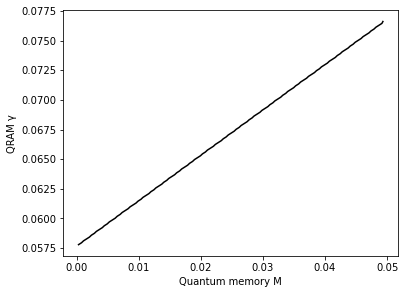}
\caption{QRAM in function of available quantum memory for minimized time.}
\label{fig:3QRACM_qspace} 
\end{figure} 

For more than $2^{0.0495d}$ quantum memory, increasing it does not improve the time complexity anymore. An important fact is that for a fixed $M$ the corresponding value $\tau_M$ from figure $\ref{fig:3time_qspace}$ and $\gamma_M$ from Figure \ref{fig:3QRACM_qspace} can be achieved simultaneously with the same algorithm.

We observe that from $M=0$ to $0.0495$ these curves are very close to affine. 
Indeed, the function that passes through the two extremities points is of expression $0.2653 - 0.1670M$. 
The difference between $\tau_M$ and its affine approximation does not exceed $4\scdot10^{-5}$. 
By the same way, the difference between $\gamma_M$ and its affine average function of expression $0.0578 + 0.3829M$ is inferior to $2 \scdot 10^{-4}$. 
All this is summarized in the following theorem. 

\begin{theorem}[Trade-off for fixed quantum memory]
	There exists a quantum algorithm using quantum random walks that solves SVP on dimension $d$ which for a parameter $M \in [0,0.0495]$ heuristically runs in time $2^{\tau_M d + o(d)}$, uses QRAM of maximum size $2^{\gamma_M d}$,  a quantum memory of size $2^{\mu_M d}$ and a classical memory of size $2^{0.2075d}$ where
	$$ \tau_M \in 0.2653 - 0.1670M + [-2\scdot10^{-5} ; 4\scdot10^{-5}]$$ 
	$$\gamma_M \in 0.0578 + 0.3829M - [0 ; 2\scdot10^{-4}] \quad ; \quad \mu_M = M.$$
\end{theorem} 

In the informal formulation of this theorem, we used the symbols $\lessapprox$ and $\gtrapprox$ that refers to these hidden small values. 

\subsection{Trade-off for fixed QRAM.} 

We also get a trade-off between QRAM and time. 
For a chosen fixed $M'$, the QRAM is denoted by $2^{\gamma_{M'} d} = 2^{M'd}$, and the corresponding minimal time by $2^{\tau_{M'} d}$. 
The required quantum memory is denoted $2^{\mu_{M'} d}$. Note that $2^{\mu_{M'} d}$ is the also the amount of the required quantum QRAM called "QRAQM".

This gives a trade-off between time and QRAM in the figure \ref{fig:4time_QRACM}, and the evolution of quantum memory in function of $M'$ is in the figure \ref{fig:4qspace_QRACM}. 

\begin{figure}[!ht]
\centering
\includegraphics[scale=0.6]{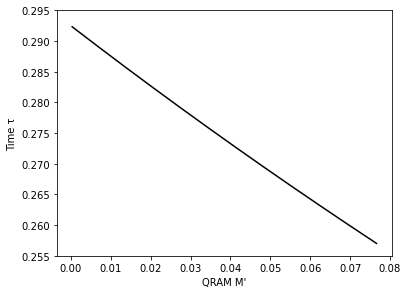}
\caption{QRAM-time trade-off.} 
\label{fig:4time_QRACM} 
\end{figure} 

\begin{figure}[!ht]
\centering
\includegraphics[scale=0.6]{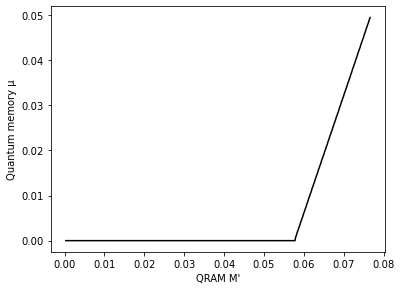}
\caption{Quantum memory in function of available QRAM for minimized time.} 
\label{fig:4qspace_QRACM} 
\end{figure} 

For more than $2^{0.0767d}$ QRAM, increasing it does not improve the time complexity. 

The difference between the function $\tau_{M'}$ and its average affine function of expression $0.2926-0.4647 \scdot M'$ does not exceed $6 \scdot 10^{-4}$. This affine function is a upper bound of $\tau_{M'}$. 

From $M'=0$ to $0.0579$ the function $\gamma_{M'}$ is at $0$. 
Then, it is close to the affine function of expression $2.6356(M'-0.0579)$. 
So $\gamma_{M'}$ can be approximated by $\max\{2.6356(M'-0.0579), 0\}$, and the difference between $\gamma_{M'}$ and this approximation does not exceed $9 \scdot 10^{-4}$. 
All this is summarized in the following theorem. 

\begin{theorem}[Trade-off for fixed QRAM]
	There exists a quantum algorithm using quantum random walks that solves SVP on dimension $d$ which for a parameter $M' \in [0,0.0767]$ heuristically runs in time $2^{\tau_{M'} d + o(d)}$, uses QRAM of maximum size $\poly(d) \cdot 2^{\gamma_{M'} d}$, a quantum memory of size $\poly(d) \cdot 2^{\mu_{M'} d}$ and uses a classical memory of size $\poly(d) \cdot 2^{0.2075d}$ where 
	$$ \tau_{M'} \in 0.2927 - 0.4647M' - [0; 6 \scdot 10^{-4}] \quad ; \quad \gamma_{M'} = M'$$
	$$\mu_{M'} \in \max\{2.6356(M'-0.0579),0\} + [0 ; 9\scdot10^{-4}].$$ 
\end{theorem}

Finally, we present a table with a few values that presents some of the above trade-offs.

\begin{figure}[H] 
    \centering
    \begin{tabular}{ |c|| c|c|c|c|c|c|c|}
    	\hline 
        Time $\tau_{M'}$ & \textbf{0.2925} & 0.2827 & 0.2733 & \textbf{0.2653} & 0.2621 & 0.2598 & \textbf{0.2570} \\
        \hline
        QRAM $\gamma_{M'}$ & \textbf{0} & 0.02 & 0.04 & \textbf{0.0578} & 0.065 & 0.070 & \textbf{0.0767} \\
        \hline
        Q. memory $\mu_{M'}$ & 0 & 0 & 0 & \textbf{0} & 0.0190 & 0.0324 & \textbf{0.0495} \\
        \hline
        Comment & {\cite{BDGL16} alg.}& & & {\cite{Laa16} alg.} & & & Thm 1. \\
        \hline
    \end{tabular}
    \caption{Time, QRAM and quantum memory values for our algorithm.}
    \label{fig:table}
\end{figure}

\section{Discussion}\label{Section:Conclusion}

\paragraph{Impact on lattice-based cryptography.} Going from a running time of $2^{0.2653d + o(d)}$ to $ 2^{0.2570 d + o(d)}$ slightly reduces the security claims based on the analysis of the SVP (usually via the BKZ algorithm). For example, if one claims $128$ bits of security using the above exponent then one must reduce this claim to $124$ bits of quantum security. This of course can usually be fixed with a slight increase of the parameters but cannot be ignored if one wants to have the same security claims as before. 

\paragraph{Parallelization.} On thing we haven't talked about in this article is whether our algorithm paralellizes well. Algorithm \ref{Algorithm:Main} seems to parallelize very well, and we argue that it is indeed the case. 

For this algorithm, the best classical algorithm takes $c \rightarrow 0$. In this case, placing each $\vv \in L$ in its corresponding $\alpha$-filters can be done in parallel and with $N$ processors (or 
$N$ width) it can be done in time $\poly(d)$. Then, there are $N$ separate instances of $\FINDALLSOLUTIONS$ which can be also perfectly parallelized and each one also takes time $\poly(d)$ when $c \rightarrow 0$. The \textbf{while} loop is repeated $N^{-\zeta} = N^{0.409d}$ times so the total running time (here depth) is $N^{0.409d + o(d)}$ with a classical circuit of width $N$. Such a result already surpasses the result from $\cite{BDGL16}$ that achieves depth $N^{1/2}$ with a quantum circuit of width $N$ using parallel Grover search. 

In the quantum setting, our algorithm parallelizes also quite well. If we consider our optimal parameters ($c = 0.3696$) with a similar reasoning, our algorithm will parallelize perfectly with $N^{1-c}$ processors (so that there is exactly one for each call to $\FINDALLSOLUTIONS$ {\ie } for the quantum random walk). Unfortunately, after that, we do not know how to parallelize well within the quantum walk. When we consider circuits of width $N$, our optimizations didn't achieve better than a depth of $N^{0.409 + o(d)}$ which is the classical parallelization. This is also the case if we use Grover's algorithm as in \cite{Laa16} for the $\FINDALLSOLUTIONS$ and we use parallel Grover search as in \cite{BDGL16} so best known (classical or quantum) algorithm with lowest depth that uses a circuit of width $N$ is the classical parallel algorithm described above. 

\section*{Acknowledgments and paths for improvements}
The authors want to thank Simon Apers for helpful discussions about quantum random walks, in particular about the fact that there are no better generic algorithms for finding $k$ different marked than to run the whole random walk (including the setup) $O(k)$ times. There could however be a smarter way to do this in our setting which would improve the overall complexity of our algorithm. Another possible improvement would be to embed the local sensitivity property in the graph on which we perform the random walk instead of working on the Johnson graph. 

\newpage
\bibliography{paper}
\bibliographystyle{alpha}
\end{document}